\documentclass[12pt]{article}
\usepackage[x11names]{xcolor}
\usepackage{graphicx,amsfonts,amssymb,epsfig,verbatim,bm,latexsym,amsmath,url,amsbsy,amsthm}
\usepackage{tocbibind,tabularx,longtable,multirow}
\usepackage{natbib}
\newtheorem{theorem}{Theorem}[section]

\newtheorem{definition}{Definition}[section]

\theoremstyle{proposition}
\newtheorem{proposition}{Proposition}[section]

\begin{document}
{\large Entropy and credit risk in highly correlated markets}\\[4mm]

\renewcommand{\thefootnote}{\fnsymbol{footnote}}
Sylvia Gottschalk\footnote[1]{Corresponding author. Middlesex University Business School, Accounting and Finance department, The Burroughs, Hendon NW4 4BT, s.gottschalk@mdx.ac.uk. We thank Sabba Hussain for exceptional research assistance.  All remaining errors are ours.}\\[10mm]
\renewcommand{\thefootnote}{\arabic{footnote}}

{\bf Abstract}: We compare two models of corporate default by calculating the Jeffreys-Kullback-Leibler divergence between their predicted default probabilities when asset correlations are either high or low. Our main results show that the divergence between the two models increases in highly correlated, volatile, and large markets, but that it is closer to zero in small markets, when asset correlations are low and firms are highly leveraged. These findings suggest that during periods of financial instability the single-and multi-factor models of corporate default will generate increasingly inconsistent predictions. \\[3mm]



\section{Introduction}

The single-factor, single-firm structural model is the cornerstone of credit risk analysis, and has thus become the building block of international financial regulation in the Basel II and III Capital Adequacy Accords \citep{BIS2006,bis2011}. In this model pair-wise asset correlations can be ignored provided the respective assets account for an asymptotically small share of a credit portfolio \citep{gordy03}. Although there is evidence that a multi-factor model would predict default risk more accurately because it would take into account asset correlations, the computational tractability of the single-factor model offsets its oversimplifying assumptions \citep{pykhtin04,emmer05,tasche06}.\\

However, asset correlations have increased significantly in the past two decades, and even more so since the 2008 financial crisis \citep[amongst others]{aitsahalia15,sandoval12,pollet10,frank09,krishnan09,morana08,duffie09,das07,rangvid01,longin95}. Higher co-movements of asset returns have negative repercussions for portfolio diversification, in which risk reduction hinges exclusively on imperfectly correlated assets. Diversification issues also arise in the risk management of portfolios of corporate debt, and of collateralised debt obligations, in which asset correlation and default correlation have been shown to be linked \citep{das06,lopez02}. Credit risk analysis in industry and financial regulation relies almost exclusively on default probabilities calculated using \citet{merton74}'s single factor structural model, which defines default as the event of the asset value of a corporate firm being lower than the face value if its debt. \citet{merton74}'s model clearly ignores asset correlations, which are introduced in extensions developed by \citet{vasicek91,vasicek02,tasche06,zhou01,eljahel04}.   \\
 
This paper builds a multi-factor model where individual corporate asset value follows a geometric Brownian motion which is correlated with that of other firms. We derive the multi-factor probability distribution of asset value and run Monte Carlo simulations assuming either high or low asset correlations\footnote{in our model correlations and volatility are either both high or both low.}. We also run simulations of the single-factor probability of default, i.e., the default probability of a corporate whose asset value is not correlated with that of other firms. The correlation matrices used to simulate the multi-factor default probability are randomly generated. Under the high correlation scenario, the pair-wise correlation coefficients are in the range $[-0.99;-0.8] \cup [0.99,0.8]$. In a low correlation matrix all the pair-wise correlation coefficients are in $[-0.4;-0.1] \cup [0.4;0.1]$. For each pair of probabilities of default, we calculate its Jeffreys-Kullback-Leibler divergence. In the context of this paper, it quantifies the extent of the misrepresentation of the actual probability of default when the single factor rather than the multi-factor model is assumed to be the correct distribution generating asset values (or vice-versa) \citep{golan06,zellner02,kullback51,burbea82}. 
Entropy measures have been extensively used in finance. The maximum entropy principle (MEP) was used to estimate the probability distribution of the underlying asset of an option using the market option prices as data in \citet{buchen96,neri12}, and to estimate the price of stock and bond options in \citet{gulko99b,gulko02}. Securities derivatives in incomplete markets can also be priced by minimizing cross-entropy as shown in \citet{branger04}. \citet{borland02,borland04} build an option-pricing model derived from a non-Gaussian model of stock returns, which are assumed to evolve according to a
nonlinear Fokker-Planck equation which maximizes the Tsallis nonextensive entropy. Finally, \citet{maasoumi02} develop an entropy metric of dependence to show -amongst other results- that stock returns are serially dependent.\\

Our results show that the probability of default of an uncorrelated firm diverges significantly from that of a multi-correlated firm when asset correlations are high, but that it diverges much less when asset correlations are low. In some cases, the divergence is close to zero and the two models may be considered proxies for each other, namely, when asset correlations and market size are low, and the debt-to-asset-value ratio is high (150 to 200 percent). However, when market size, debt-to-value ratios \emph{and} correlations are high, the two models produce inconsistent default probabilities. Clearly, as asset correlations rise the single-factor model starts to misrepresent actual default probabilities. Further cases are analysed in Section \ref{results}.\\

Overall, we find that the discrepancy between the two models is exacerbated when firm's indebtedness is between 10 to 100 percent and in highly correlated markets. Our findings have implications for financial regulation. The Basel II and III Capital Adequacy Accords stipulate that capital provisions should be calculated in accordance with a single-factor, single-firm structural model \citep{BIS2006,bis2011}. Our paper suggests that in periods of financial instability, when asset volatility and correlations increase, one of the models may misreport default risk and thus lead to inadequate capital provisions. Our results are congruent with \citet{aitsahalia15,pollet10,krishnan09}. These papers find evidence that asset correlations have more predictive power than asset volatility. Further, \citet{das06} found that clustering of defaults occurs during times of high volatility because both default probabilities and correlation between defaults increase.\\ 

The paper proceeds as follows: Section \ref{mmerton} presents our main model and Section \ref{divergence} measures of divergence. Simulations of high/low correlation matrices are described in Section \ref{vine}. Section \ref{results} presents and discusses the results of simulations of our model. Section \ref{conc} concludes and precedes the Appendix. 

\section{Multi-factor model of corporate value\label{mmerton}}

Consider $m$ dependent Brownian motions $\bm{B_t}=(B_{1,t},...,B_{m,t})$ defined on a probability space $(\Omega,\mathcal{F},\mathbb{P} )$, where ${\cal{F}}(t)$ is the filtration associated with $\bm{B_t}$, and $\mathbb{P}$ a probability measure. The  market has $N$ firms with correlated asset values $\{V_{1,t},...,,V_{N,t}\}$, for $t\in [0,T]$. $V_{i,t}$ satisfies the differential equation 

\begin{equation}
dV_{i,t}= V_{i,t}\left(\mu_{i}dt+ \sigma_{i}dB_{i,t}\right)\label{vdep} 
\end{equation}

where $\mu_{i}$ is its drift and $\sigma_i$ its volatility. The Brownian motions $B_{j,t}$ are correlated, in the sense that $dB_{jt}dB_{kt}=\rho_{jk}$, for $j,k=1,...,m$, with correlation coefficient $\rho_{jk}= \frac{1}{\sigma_{j}\sigma_{l}}\sum_{l=1}^{m} \sigma_{jl}\sigma_{kl}$, for all $j,k=1,...,m$.\\

$B_{j,t}$, $j=1,...,m$ can be rewritten as a function of independent Brownian motions $W_{j,t}$. Let 

\begin{equation}
\sigma_{i} =\left[\sum_{j=1}^{m} \sigma_{ij}^{2}\right]^{1/2}\label{sigma1}
\qquad \textrm{for all i=1,...,N}
\end{equation}

and 

\begin{equation}
dB_{j,t}=\sum_{j=1}^{m} \frac{\sigma_{ij}}{\sigma_{i}}dW_{j,t}\label{db1}
\qquad \textrm{for all i=1,...,N}
\end{equation}

then, by substituting in (\ref{sigma1}) and (\ref{db1}) in (\ref{vdep}), (\ref{vdep}) becomes

\begin{equation}
dV_{i,t}= V_{i,t}\left(\mu_{i}dt+ \sum_{j=1}^{m} \sigma_{ij}dW_{j,t}\right)\label{gbm} 
\end{equation}

We show below that the volatility of $V_{i,t}$, $i=1,..,N$, depends on the pair-wise covariances between the $m$ stochastic processes $W_{j,t}$ driving the value of the other firms in the market.\\

\begin{proposition} 
\label{distribution}
The solution of (\ref{gbm}) is the set of values $\{V_{1,t},...,,V_{N,t}\}$ such that

\begin{equation}
V_{i,t}= V_{i,0}exp\left\{(\mu_{i}-\frac{1}{2}\sum_{j=1}^{m}\sigma_{ij}^2)t+ 
 \sum_{j=1}^{m}\sigma_{ij}W_{j,t} \right \}\qquad i=1,...,N\label{solsystem}
\end{equation}

where $exp\{.\}\equiv e^{(.)}$ is the exponential function. Moreover, $V_{i,t}$ is log-Normally distributed with mean

\begin{equation}
E[V_{i,t}]=V_{i,0}\exp\{(\mu_i) t\}
\end{equation} 

and variance 
\begin{equation}
Var[V_{i,t}]=V_{i,0}^2\exp\{2\mu_i t\}(\exp\left \{\sum_{j=1}^{m}\sigma_{ij}^{2}t \right\}-1)
\end{equation}
\end{proposition}

The proof can be found in Appendix A. \\

Default is defined as the event of the value of the $i$-th firm, $V_{i,T}$, being lower than the face value of its debt $D_i$ at a given time $t=T$. The probability of this event is

\begin{equation}
P\left (V_{i,T}\leq D_i\right)=P\left(V_{i,0}X_{i,T}\leq D_i\right)=
P\left(X_{i,T}\leq \ln(D_i/V_{i,0}) \right)\qquad i=1,...,N\label{pd1}
\end{equation}

where $X_{i,t}=\exp\{(\mu_{i}-\frac{1}{2}\sum_{j=1}^{m}\sigma_{ij}^2)t+\sum_{j=1}^{m}\sigma_{ij}W_{j,t}\}$ for i=1,...,N.

\begin{equation}
P\left (X_{i,T}\leq \ln(D_i/V_{i,0})\right)=N_m\left(\ln(D_i/V_{i,0})\right)\qquad i=1,...,N \label{probm}
\end{equation}

where $N_m(.)$ is the cumulative Normal distribution of $X_{i,t}$ when $m>1$.\\

The single-factor model is retrieved from (\ref{gbm}) and (\ref{solsystem}) for $m=1$, 

\begin{equation}
V_{i,t}=V_{i,0} exp\left\{\left(\mu_{i}-\frac{1}{2}(\sigma_{i})^{2}\right)t+\sigma_{i}W_{i}\right \}\qquad i=1,...,N\label{solsingle}
\end{equation}

and 

\begin{equation}
P\left (V_{i,T}\leq D_i\right)=N_s(\ln(D_i/V_{i,0})) \qquad i=1,...,N\label{probs}
\end{equation}

where $N_s(.)$ is the cumulative Normal distribution of $X_{i,t}$ when $m=1$.\\

The model presented above is closely related to the single-and multi-factor models used extensively in credit risk analysis \citep{tasche06,pykhtin04,vasicek02}. In these papers, pairwise correlations between two assets are not directly modelled. It is assumed instead that both assets are correlated to one (or more) underlying factor(s), and that the assets are independent from each other \emph{conditional} on the common factor(s). Their set-up can be retrieved from our model by assuming that each Brownian motion $B_{i,t}= \sqrt{\rho}Y+ \sqrt{1-\rho}\epsilon_i$, $i=1,..,m$, where $Y$ is the (single) common factor, and $\epsilon_i$ is a firm-specific error, rather than $dB_{jt}dB_{kt}=\rho_{jk}$, for $j,k=1,...,m$ as we do above. 

\section{Divergence measures and simulation of correlation matrices\label{divergence}}

We examine the discrepancy between (\ref{probs}) and (\ref{probm}) by measuring the Jeffreys-Kullback-Leibler divergence between these two probabilities. Let $f_1(x)$ and $f_2(x)$ be two probability distribution functions over $\mathbb{R}^m$, where $m$ is the number of dependent Brownian motions in Section \ref{mmerton} above.\\

The Jeffreys-Kullback-Leibler divergence measure between $f_1(x)$ and $f_2(x)$ is defined as

\begin{equation}
J(f_1,f_2)= I_{1,2}(f_1,f_2)+I_{2,1}(f_2,f_1) \label{jkl0}
\end{equation}

where 

\begin{equation}
I_{1,2}(f_1,f_2)=\int_{0}^{\infty} f_1(x) log \left(\frac{f_1(x)}{f_2(x)}\right)dx\label{cross_ent}
\end{equation}

is the Kullback-Leibler divergence or cross-entropy\footnote{(\ref{cross_ent}) reduces to the Shannon entropy when $f_2(x)$ is a uniform distribution, i.e. $-\int_{0}^{\infty} f_1(x) log(f_1(x))dx$.}. \citep{kullback51}. Although (\ref{cross_ent}) does satisfy $I_{1,2}(f_1,f_1) = 0$ and the positivity condition $I_{1,2}(f_1,f_2) > 0$ whenever $f_1\neq f_2$, it is not a true metric distance, because it is not symmetric and does not satisfy the triangle inequality \citep{ullah96}. However, it can be thought of as an ``entropy distance" between $f_1(x)$ and $f_2(x)$. It quantifies the loss of information occuring when considering $f_2(x)$ as the correct probability distribution when the true distribution is $f_1(x)$. If we were to use the Kullback-Leibler divergence measure, we would have to arbitrarily decide whether the multi-factor or the single-factor probability of default is the true probability of default. Owing to the asymmetry of the Kullback-Leibler divergence, if the multi-factor distribution $f_{m}$ is assumed to be the true distribution, then $I_{m,s}(f_{m},f_{s})=\int_{0}^{\infty} f_{m}(x) log \left(\frac{f_{m}(x)}{f_{s}(x)}\right)dx$, which differs from $I_{s,m}(f_{s},f_{s})=\int_{0}^{\infty} f_{s}(x) log \left(\frac{f_{s}(x)}{f_{m}(x)}\right)dx$.\\

In this paper, we restrict ourselves to measuring the divergence between the probability of default resulting from two different models of corporate default without making \emph{a priori} assumptions about the accuracy of any one model. Consequently, we opt for a symmetric extension of the Kullback-Leibler measure proposed in \citet{kullback51} and \citet{burbea82}, the Jeffreys-Kullback-Leibler divergence measure.\\
 
Substituting (\ref{cross_ent}) into (\ref{jkl0}), (\ref{jkl0}) becomes

\begin{equation}
J(f_1,f_2)=\int_{0}^{\infty} (f_1(x)-f_2(x)) log \left(\frac{f_1(x)}{f_2(x)}\right)dx\label{jeffreys}
\end{equation}

(\ref{jeffreys}) remains a pseudo-metric, since it violates the triangle inequality, but it does satisfy all the other properties of a metric \citep{kullback51,ullah96}.

\subsection{Simulating correlation matrices\label{vine}}

The most straightforward way to simulate a random correlation matrix consists of generating random data from a given distribution and then calculate their pair-wise correlations. However, the resulting matrix may not necessarily have the desired consistently high (or low) pairwise correlations. The simulated correlation structure should also be realistic. A deterministic matrix in which all off-diagonal entries are equal to 0.9 would fulfill the requirement of having a high correlation structure. It would nonetheless be a very poor proxy for the asset correlations found in most financial markets. In order to systematically generate \emph{realistic} positive definite matrices with given correlations, we adapt an algorithm created by \citet{hardin13}. In their paper, a ``noise" is added to a correlation matrix in such a way that the resulting matrix will have blocks of pair-wise correlations whose values are within a predetermined range.\\

More precisely, let $S=(S_{ij})_{i,j=1}^m$ be a $m \times m$ matrix generated by calculating the pairwise correlations of random data. From a Uniform probability distribution we draw a ``noise" $\delta_{ij} \in [-1,1]$ which is added to $S_{ij}$, $i\neq j$, subject to the restriction that the resulting entry $\hat{S}_{ij}=S_{ij}+\delta_{ij}$, $i\neq j$, should be between $[\rho_{min},\rho_{max}]$. In our ``high asset correlation" scenario, all the entries $\hat{S}_{ij}$ are within $[-0.99,-0.8]$ or within $[0.8,0.99]$. In the ``low asset correlation" alternative, all the entries $\hat{S}_{ij}$ are within $[-0.4,-0.1]$ or within $[0.1,0.4]$. \\
   
Our main modification to \citet{hardin13}'s algorithm is that we allow $\delta_{ij}$, $\rho_{min}$ and $\rho_{max}$ to be negative, since financial asset correlations can be either positive or negative. In their paper correlations are always positive.

\section{Results and discussion\label{results}}

We ran 2000 Monte Carlo simulations of the multi-factor and single factor probabilities of default\footnote{All our simulations were run in R \citep{r16}. The algorithm that generates random correlation matrices was adapted from \citet{hardin13} and \citet{joe06}.}, for a given number of firms $N=$10,50,90,100,500, and 1000\footnote{For the sake of simplicity, we set $N=m$ in the simulations.}, and a given level of debt leverage $ln(D_i/V_i)$=0.1,.., 2, in steps of 0.1. For each simulation, we randomly estimate a high correlation $N\times N$ matrix. This procedure is repeated for low correlation matrices. The results of each set of simulations are 2000 values for the Jeffreys-Kullback-Leibler divergence (\ref{jeffreys}), from which we calculate the averages shown in Figure 1 and Table \ref{tab:J}. Our main results are summarised as follows. \\

\emph{The divergence between the multi-factor and the single-factor probability of default increases with asset correlation.}\\

For $N$=10 to 100 firms, the graph of the low correlation divergence is significantly below that of the high correlation divergence. For $N=10$ firms, leverage equal to 10\%, and low correlation, the average divergence is $\bar{J}^L$=0.3400. For the same number of firms and leverage, but high correlation, $\bar{J}^H$=0.5669. The divergences for a leverage of 200\% -all else equal- are $\bar{J}^L$= 0.0575 and $\bar{J}^H$=0.1377. Table 1 shows analogous results for other levels of leverage and market sizes. This suggests that in periods of high correlation, e.g., financial crises, the multi-factor default probability will differ significantly from that of the single-factor model. One of the two models may underestimate the true default probability. For larger market sizes, $N$=500 and 1000, the discrepancy between the high and low correlation divergences is reduced. This issue is addressed in detail below. \\

\emph{The divergence between the multi-factor and the single-factor probability of default is inversely related to leverage.} \\

This result indicates that the more indebted the firm, the less discrepancy between the information given by the two models. For $N$=50 and low correlation, $\bar{J}^L$=0.5834 if $ln(D_i/V_i)$=10\%,  whilst $\bar{J}^L$=0.1018 for $ln(D_i/V_i)$=200\%, all else equal. Clearly, in the case of highly indebted firms, default probability tends to be similar under either model,. Figure 1 suggests that irrespective of asset correlations, the divergence between the two probabilities tends to zero when debt leverage increases. The implication is that for lower levels of leverage, the two probabilities of default will diverge considerably, and will be influenced by asset correlations. For instance, it can be seen in Table \ref{tab:J} that for a leverage equal to 50\%, high correlation and $N$=50, the average divergence is $\bar{J}^H$=0.3498, against $\bar{J}^L$=0.2636. At that level of leverage, the multi-factor and single factor probabilities of default will give contradictory signals.\\

\begin{figure}
\begin{center}
\caption{Divergence by market size (N=10 to 1000 firms)\label{c1}}
\includegraphics[totalheight=15cm,width=12cm]
{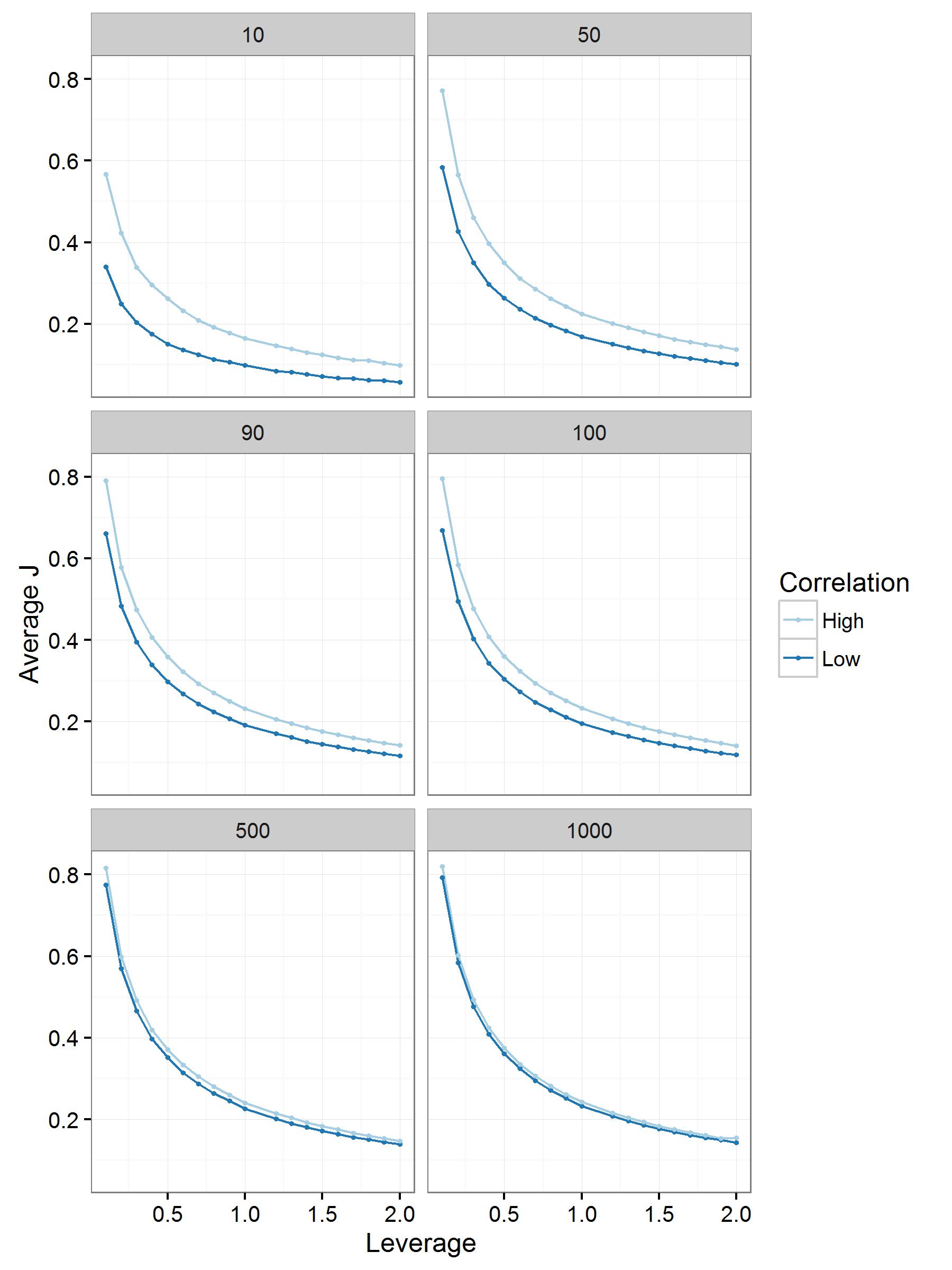}
\end{center}
\end{figure}

\begin{table}[ht]
  \begin{center}
  \caption{Average values of Jeffreys-Kullback-Leibler divergence }
    \begin{tabular}{ccccccccc}
    \hline
  Leverage $\downarrow$  & \multicolumn{4}{c}{Low correlation ($\bar{J}^L$)}  & \multicolumn{4}{c}{High correlation ($\bar{J}^{H}$)} \\  \hline
   &N=10    &50    & 100   & 1000  & 10    & 50    & 100   & 1000 \\
    \multicolumn{1}{c}{\multirow{2}[0]{*}{0.1}} & 0.3400 & 0.5834 & 0.6683 & 0.7926 & 0.5669 & 0.7722 & 0.7952 & 0.8201 \\
    \multicolumn{1}{c}{} & (0.1268) & (0.0813) & (0.0658) & (0.0241) & (0.1982) & (0.1081) & (0.0763) & (0.0248) \\
    \multicolumn{1}{c}{\multirow{2}[0]{*}{0.5}} & 0.1504 & 0.2636 & 0.3046 & 0.3610 & 0.2618 & 0.3498 & 0.3607 & 0.3751 \\
    \multicolumn{1}{c}{} & (0.0567) & (0.0382) & (0.0308) & (0.0115) & (0.0931) & (0.0488) & (0.0358) & (0.0118) \\
    \multicolumn{1}{c}{\multirow{2}[0]{*}{1}} & 0.0980 & 0.1684 & 0.1954 & 0.2328 & 0.1651 & 0.2243 & 0.2330 & 0.2425 \\
    \multicolumn{1}{c}{} & (0.0364) & (0.0246) & (0.0201) & (0.0074) & (0.0608) & (0.0318) & (0.0229) & (0.0073) \\
    \multicolumn{1}{c}{\multirow{2}[0]{*}{1.5}} & 0.0710& 0.1266 & 0.1475 & 0.1772 & 0.1240 & 0.1700 & 0.1764 & 0.1833 \\
    \multicolumn{1}{c}{} & (0.0275) & (0.0186) & (0.0140) & (0.0056) & (0.0448) & (0.0239) & (0.0172) & (0.0056) \\
    \multicolumn{1}{c}{\multirow{2}[0]{*}{2}} & 0.0575 & 0.1018 & 0.1185 & 0.1425 & 0.0992 & 0.1377 & 0.1415 & 0.1547 \\
    \multicolumn{1}{c}{} & (0.0215) & (0.0147) & (0.0119) & (0.0040) & (0.0365) & (0.0192) & (0.0134) & (0.0052) \\  \hline
    \end{tabular}
		\end{center}
		{\footnotesize 2000 simulations of correlation matrices and probability of default. Standard errors in brackets.}   
  \label{tab:J}%

 \begin{center}
  \caption{Welch test of equality of $\bar{J}^L$ and $\bar{J}^H$ in Table 1}
    \begin{tabular}{ccccccccc}
    \hline
 Leverage $\downarrow$& \multicolumn{2}{c}{N=10} & \multicolumn{2}{c}{50} & \multicolumn{2}{c}{100} & \multicolumn{2}{c}{1000} \\
    \hline
          & t     & p-value & t     & p-value & t     & p-value & t     & p-value \\
    0.1   & -65.414 & 2.20E-16 & -62.447 & 2.20E-16 & -56.398 & 2.20E-16 & -11.241 & 2.20E-16 \\
    0.5   & -45.704 & 2.20E-16 & -62.228 & 2.20E-16 & -53.087 & 2.20E-16 &  -14.13  & 0.00 \\
    1     & -42.279 & 2.20E-16 & -62.182 & 2.20E-16 & -56.788 & 1.31E-11 & -13.27 & 2.20E-16 \\
    1.5   & -45.062 & 2.20E-16 & -65.491 & 2.20E-16 & -56.619 & 2.20E-16 & -10.863 & 2.20E-16 \\
    2     & -44.008 & 2.20E-16 & -66.428 & 2.20E-16 & -56.152 & 2.20E-16 & -26.415 & 2.20E-16 \\
    \hline
    \end{tabular}%
		\end{center}
			{\footnotesize $H_0:\bar{J}^L-\bar{J}^H=0$ and $H_1:\bar{J}^L-\bar{J}^H \neq 0$. Equal sample sizes.} 
  \label{tab:welch}%
\end{table}%

\emph{Irrespective of asset correlations, the difference between $\bar{J}^H$ and $\bar{J}^L$ diminishes with the number of firms.}\\

For larger market sizes, $N$=500 and 1000 firms, the difference between the high/low correlation divergences diminishes, and even seems to disappear in Figure 1. Table 1 shows that for $ln(D_i/V_i)$=150\%, and $N$=1000 firms, $\bar{J}^L$= 0.1772 whilst $\bar{J}^H$= 0.1833. However, this difference is statistically significant at a 1\% confidence level, as evidenced by the Welch test of equality of means shown in Table 2. The Welch test is also significant at 1\% confidence level for all other market sizes and leverage values. It should be emphasized that we are not suggesting that the multi-factor and single-factor default probabilities will converge or even equalize when the number of firms increases. As long as $\bar{J}\neq 0$, the two probabilities of default will diverge. Tables 1 and 2 support this finding even in large markets.  \\

\emph{The divergence between the two models increases with market size.}\\ 

This can be seen in Figure 1 by comparing the value of average divergence when $N=10$ and $N=1000$. For $N=10$, high correlation and leverage equal to 10\%, $\bar{J}^H$= 0.5669, against $\bar{J}^H$= 0.8201 for $N=1000$, all else equal. When correlations are low and leverage equal to 10\%, $\bar{J}^L$= 0.3400 for $N=10$ and $\bar{J}^L$= 0.7929 for $N=1000$. The level of firm indebtedness does not alter this result. For $N=10$ and leverage at 200\%, both high and low correlation divergences become closer to zero, $\bar{J}^L$=0.0575 and $\bar{J}^H$=0.0992. For $N=1000$, and leverage at 200\%, high and low correlation divergences become closer to each other in value, $\bar{J}^L$=0.1425 and $\bar{J}^L$=0.1547, but they are still markedly higher than the respective divergences for $N=10$.

\section{Conclusion\label{conc}}

This paper builds a multi-factor model where individual corporate asset value follows a geometric Brownian motion which is correlated with that of other firms. We run simulations of the multi-factor default probability using randomly generated low and high correlation matrices. We also run simulations of the single factor default probability, i.e., the default probability of a corporate with uncorrelated asset value. For each pair of default probabilities we calculate its Jeffreys-Kullback-Leibler divergence. Overall, we find that the discrepancy between the two models is exacerbated in highly correlated markets and when firm's indebtedness is between 10 to 100 percent. Our findings have implications for financial regulation. The Basel II and III Capital Adequacy Accords stipulate that capital provisions should be calculated in accordance with a single-factor, single-firm structural model \citep{BIS2006,bis2011}. Our paper suggests that in periods of financial instability, when asset volatility and correlations increase, one of the models may misreport default risk and thus lead to inadequate capital provisions.

\bibliography{structural,entropy,correlation,regulation}
\bibliographystyle{apecon}

\appendix
\section{Appendix}

Our proof hinges on the Brownian motion $W_{i,t}$ being Normally distributed with mean 0 and variance 1, and independent from $W_{j,t}$. It also relies on the $1/2$-self-similarity property of Brownian motions, i.e., $t^{1/2}W_{j,t}{\buildrel d \over =} W_{j,t}$. Before we proceed to the proof of Proposition \ref{distribution}, a couple of definitions and a theorem will be presented. A good introduction on self-similar stochastic processes is \citet{maejima02}.

\subsection{Self-similarity of Brownian motions}

A stochastic process $\{W_{t}, t\geq 0\}$ is said to have independent increments,
if for any $m\geq 1$ and for any partition $0\leq t_0 < t_1 <..< t_m$, $W_{t_1}- 
W_{t_0},.., W_{t_m}-W_{t_{_{m-1}}}$ are independent. A formal definition of a Brownian motion will be useful in the proof of Theorem \ref{ss2}.  

\begin{definition}\label{def_bm} If a stochastic process $\{W_{t}, t\geq 0\}$ satisfies
\begin{enumerate}

\item[(i)]  $W_0 = 0$ a.s.,
\item[(ii)] it has independent and stationary increments,
\item[(iii)] for each $t>0$, $W_t$ has a Gaussian distribution with mean zero and variance $t$, and
\item[(iv)] its sample paths are continuous a.s.,
\end{enumerate}

then it is called (standard) Brownian motion.
\end{definition}

\begin{definition} A stochastic process $\{W_{t}, t\geq 0\}$ is said to be self-similar
if for any $a>0$, there exists $b>0$ such that\\
\begin{equation} 
W_{at}\buildrel d \over = b W_{t}\label{ss1}
\end{equation}
\end{definition}

where ``${\buildrel d \over =}$'' means equality in distribution. \\

The following theorem was proved by \citet{lam62} and shows that a self-similar stochastic process $W_{at}$ is equal in distribution to the stochastic process $a^{H}W_{t}$. In many texts on the topic, a self-similar process is defined by this property.  
 
\begin{theorem}(\citet{lam62}) If $\{W_{t}, t\geq 0\}$ is nontrivial, stochastically continuous at t=0 and self-similar, then there exists a unique exponent $H\geq 0$ such that b in (\ref{ss1}) can be expressed as $b = aH$. Moreover, $H>0$ if and only if $W_0 = 0$ a.s.
\end{theorem}

The following theorem will be used in our proof of Proposition \ref{distribution}.

\begin{theorem} A Brownian motion $\{W_{t}, t\geq 0\}$ is 1/2-self-similar.\label{ss2}\end{theorem}

\begin{proof} It is sufficient to show that for every $a > 0$, $W_{at}$ is also Brownian
motion as defined in \ref{def_bm} above.\\

Conditions (i), (ii) and (iv) are trivially fulfilled. Regarding (iii), Gaussianity and mean-zero property also follow from the properties of $W_{t}$. To obtain the variance, consider the stochastic process ${a^{-1/2}W_{at}, t\geq 0}$. Its variance is 

\begin{equation}
E[(a^{-1/2}W_{at})^{2}] = a^{-1}E[W_{at}^{2}]= a^{-1}Var[W(at)]=a^{-1}(at)=t
\end{equation}

Thus $W_{at}$ is a Brownian motion.
\end{proof}

\subsection{Proof of Proposition \ref{distribution}}

The proof that (\ref{solsystem}) is the solution of (\ref{gbm}) can be found in \citet{etheridge01} or \citet{shreve2}, amongst others. The distributional properties of $V_{i,t}$ are not proved in the literature, due to it being the building block of option-pricing models, rather than a stand-alone model as in this paper.\\
 
\begin{proof}
Let $X_t= \exp\{(\mu_i-\frac{1}{2} \sum_{j=1}^{N} \sigma_{i,j}^2)t+\sum_{j=1}^{N} \sigma_{i,j}W_{t,j}\}$. Then $V_{i,t}=V_{i,0}X_{i,t}$, where $V_{i,0}$ is a constant. The {\bf mean of $X_t$} is

\begin{equation}
E[X_t]=\exp\left\{(\mu_i-\frac{1}{2} \sum_{j=1}^{N} \sigma_{i,j}^2)t\right\} E\left[\exp\left\{\sum_{j=1}^{N} \sigma_{i,j}W_{j,t}\right\}\right]\label{e1}
\end{equation}

since the terms in the first exponential are non-random. Given that the $m$ Brownian motions $W_{1,t},...,W_{m,t}$ are independent, the term in the expectation function can be written as 

\begin{equation}
E\left[\exp\left\{\sum_{j=1}^{m} \sigma_{i,j}W_{j,t}\right\}\right]= \prod_{j=1}^{m}E\left[\exp \left\{\sigma_{i,j}W_{j,t}\right\}\right]
\end{equation}

From the $1/2$-self-similarity property,

\begin{equation}
\prod_{j=1}^{m}E\left[\exp \left\{\sigma_{i,j}W_{j,t}\right\}\right]=\prod_{j=1}^{m}E\left[\exp \left\{\sigma_{i,j}t^{1/2}W_{t,j}\right\}\right]
\end{equation}
 
Since $W_{j,t}$ is N(0,1), $E[e^{a W_{j,t}}]=e^{a^2/2}$, where $a$ is a real constant. This implies that

\begin{equation}
\prod_{j=1}^{m}E\left[\exp \left\{\sigma_{i,j}t^{1/2}W_{j,t}\right\}\right]=
\prod_{j=1}^{m}\exp \left\{\frac{1}{2}\sigma_{i,j}^2 t \right\}\label{e2}
\end{equation}

Substituting in (\ref{e2}) in (\ref{e1}) yields

\begin{equation}
E[X_t]=\exp\{\mu_i t\} 
\end{equation}\vspace{5mm}

The {\bf variance of $X_t$}, $Var(X_t)=E[X_t^2-(E(X_t)^2)]$, can be simplified to 

\begin{equation}
Var[X_t]=\exp\left\{2\mu_i t-\sum_{j=1}^{m} \sigma_{i,j}^2 t\right\}E\left[\exp\left\{2\sum_{j=1}^{m} \sigma_{i,j}W_{j,t}\right\}\right]-\exp\{2\mu_i t\}
\end{equation}

From the 1/2-self-similarity of Brownian motion

\begin{equation}
E\left[\exp\left\{2\sum_{j=1}^{m} \sigma_{i,j}W_{j,t}\right\}\right]=\prod_{j=1}^{m}E\left[\exp \left\{2\sigma_{i,j}t^{1/2}W_{j,t}\right\}\right]=
\prod_{j=1}^{m}\exp \left\{2\sigma_{i,j}^2 t \right\}\label{v2}
\end{equation}

Substituting in above 

\begin{equation}
Var[X_t]=\exp\left\{2\mu_i t-\sum_{j=1}^{N} \sigma_{i,j}^2 t+2\sum_{j=1}^{N}\sigma_{i,j}^2 t \right\}-\exp\{2\mu_i t\}
\end{equation}

From which,

\begin{equation}
Var[X_t]=\exp\{2\mu_i t\}\left(\exp\left\{\sum_{j=1}^{m}\sigma_{i,j}^2 t\right\}-1\right)
\end{equation}

$V_{i,t}= V_{i,0}X_t$. Since $V_{i,0}$ is a constant, the mean and variance of $V_{i,t}$ follow trivially. \\
\end{proof}

\end{document}